\documentclass{amsart}

\usepackage[T1]{fontenc}
\usepackage{gastex}
\usepackage[english]{babel}
\usepackage{fullpage}
\usepackage{graphicx}
\usepackage{color}

\newtheorem{proposition}{Proposition}[section]
\newtheorem{theorem}[proposition]{Theorem}
\newtheorem{lemma}[proposition]{Lemma}
\newtheorem{corollary}[proposition]{Corollary}
\newtheorem{definition}[proposition]{Definition}

\newtheorem{counterexample}[proposition]{Counterexample}

\begin{document}
\title[Periodic configurations of subshifts on groups]{Periodic configurations of subshifts on groups}
\author[F. Fiorenzi]{Francesca Fiorenzi}

\address{Francesca Fiorenzi, {\rm Laboratoire de Recherche en Informatique, UMR 8623, Université Paris-Sud, 91405 Orsay, France}}

\email{fiorenzi@lri.fr}

\begin{abstract}
We study the density of periodic configurations for shift spaces
defined on (the Cayley graph of) a finitely generated group. We
prove that in the case of a full shift on a residually finite group
and in that of a group shift space on an abelian group, the periodic
configurations are dense. In the one--dimensional case we prove the
density for irreducible sofic shifts. In connection with this we
study the surjunctivity of cellular automata and local selfmappings.
Some related decision problems for shift spaces of finite type are
also investigated.

\end{abstract}

\maketitle

\section{Introduction}
In symbolic dynamics, a shift space is a set of bi--infinite words
over a finite alphabet which avoid a fixed set of forbidden factors.
It is so called because of its invariance under the shift map. A
shift is an example of a discrete--time dynamical system, i.e. a
compact metric space equipped with a continuous selfmapping that
describes one step of the evolution. In this framework, it is
interesting to study the behavior of points and sets under
iteration. In particular, some typical questions concern the density
of periodic points and the topological transitivity -- the latter
corresponding to the so--called irreducibility of the shift space.

A continuous map between two shifts that commutes with the shift map
is called a local function. A bijective local function is called a
conjugacy. An important open question in symbolic dynamics is to
decide whether two shifts are conjugate, even when they are of
finite type (i.e. described by a finite set of forbidden factors).

\medskip

In this work we consider a more general class of shift spaces, in
which instead of bi--infinite words we consider tilings of a
suitable regular graph avoiding some forbidden patterns. More
precisely, let $\Gamma$ be a finitely generated group represented by
its Cayley graph. A configuration is an element of the space
$\mathcal{A}^\Gamma$, i.e. a tiling of $\Gamma$ by means of letters
in a finite alphabet $\mathcal{A}$.
%
%The space $\mathcal{A}^\Gamma$ is naturally endowed with a metric
%and hence with an induced topology, this latter being equivalent to
%the usual product topology, where the topology in $\mathcal{A}$ is
%the discrete one. A subset $X$ of $\mathcal{A}^\Gamma$ which is
%$\Gamma$--invariant and topologically closed is called subshift,
%shift space or simply shift. This topological definition of shift is
%equivalent to the combinatorial one
%
A subset $X$ of $\mathcal{A}^\Gamma$ whose configurations avoid a
fixed set of forbidden patterns is called subshift, shift space or
simply shift. An $n$--dimensional shift is a subshift defined on the
group $\mathbb{Z}^n$. It is clear that one--dimensional shifts are
suitable subsets of bi--infinite words over a finite alphabet. Also
in this case a shift is naturally endowed with a compact metric and
shift maps in direction of each neighbor are as many continuous
selfmappings. On these shift spaces is still possible to give the
notion of local function. A cellular automaton is a local function
defined on the whole $\mathcal{A}^\Gamma$ (the full
$\mathcal{A}$--shift).

In this setting we prove that the density of the periodic
configurations is a conjugacy invariant, as is the number of
periodic configurations of a fixed period. Moreover, we show that a
group $\Gamma$ is residually finite if and only if periodic
configurations are dense in $\mathcal{A}^\Gamma$. If the alphabet
$\mathcal{A}$ is a finite group and $\Gamma$ is abelian, then
periodic configurations of a subshift which is also a subgroup of
$\mathcal{A}^{\Gamma}$ (namely a group shift) are dense.  In the
one--dimensional case, we prove the density of the periodic
configurations for an irreducible subshift of finite type of
$\mathcal{A}^\mathbb{Z}$. A sofic shift being the image under a
local map of a shift of finite type, this implies also the density
of the periodic configurations for an irreducible sofic subshift of
$\mathcal{A}^\mathbb{Z}$. We see that these results cannot be
generalized to higher dimensions.

We also investigate the surjunctivity of local selfmappings (a
selfmapping is surjunctive if it is either non--injective or
surjective). Richardson has proved in~\cite{Richardson72} that an
$n$--dimensional cellular automaton is surjunctive. In fact, the
surjunctivity problem is related to that of the density of periodic
configurations: we prove that if the periodic configurations of a
shift are dense, then its local selfmappings are surjunctive. As a
consequence we have that a cellular automaton on
$\mathcal{A}^\Gamma$ is surjunctive if $\Gamma$ is a residually
finite group.

Cellular automata have mainly been investigated in the
$n$--dimensional case. There is a deep difference between the
one--dimensional cellular automata and the higher dimensional ones.
For example, Amoroso and Patt have shown in \cite{AmorosoPatt72}
that surjectivity and injectivity of one--dimensional cellular
automata are decidable. On the other hand Kari has shown in
\cite{Kari90} and \cite{Kari94} that both the injectivity and the
surjectivity problems are undecidable for cellular automata of
higher dimension. In this work we extend the Amoroso--Patt's results
to local functions defined on shifts of finite type. Some other
well--known decision problems for $n$--dimensional shifts of finite
type are listed, proving that in the one--dimensional case they can
be solved. More generally they can be solved for the class of group
shifts using some results due to Wang \cite{Wang61} and Kitchens and
Schmidt \cite{KitchensSchmidt88}.

\medskip

The paper is organized as follows. In Sections~\ref{Cayley graphs of
finitely generated groups} and~\ref{Shift spaces and cellular
automata} the notions of shift space and local function are formally
defined, also proving that many basic results for the subshifts of
$\mathcal{A}^\mathbb{Z}$ given in the book of Lind and Marcus
\cite{LindMarcus95}, can be generalized to the subshifts of
$\mathcal{A}^\Gamma$.

In Section~\ref{One--dimensional shifts spaces}, we recall some
relevant classes of one--dimensional shifts and we give our
extension of the Amoroso--Patt's decidability results.

In Section~\ref{Density of periodic configurations}, after some
generalities about periodic configurations, we prove that their
density is a conjugacy invariant.

In Section~\ref{Residually finite groups} we recall the class of
residually finite groups and we prove that a group $\Gamma$ is
residually finite if and only if periodic configurations are dense
in $\mathcal{A}^\Gamma$.

In Section~\ref{Group shifts} we prove the density of periodic
configurations for group shifts.

In Section~\ref{Surjunctivity}, we study the surjunctivity of a
general cellular automaton on a group $\Gamma$. We also mention the
related ``Garden of Eden'' theorem which, whenever holds, guarantees
the surjunctivity of local selfmappings.

Section~\ref{The n-dimensional case} is devoted to establish for
which classes of $n$--dimensional shifts the periodic configurations
are dense. We conclude by listing some other well--known decision
problems for this class of shifts.

\section{Cayley graphs of finitely generated groups}\label{Cayley graphs of finitely generated groups}
Let $\Gamma$ be a finitely generated group for which we fix a finite
set of generators. Without loss of generality we suppose that it is
symmetric (the inverse of a generator is still a generator).

, we fix Each $\gamma \in \Gamma$ can be written as
\begin{equation}\label{decomposition}
\gamma = x_{i_1}^{\delta_1}x_{i_2}^{\delta_2} \dots
x_{i_n}^{\delta_n}
\end{equation}
where the $x_{i_j}$'s are generators and $\delta_j \in \mathbb{Z}$.
We define the \emph{length of $\gamma$} as the natural number $$\|
\gamma \| = \min \{ |\delta_1| + |\delta_2| + \dots + |\delta_n|
\mid {\rm \gamma\ is\ written\ as\ in\ (\ref{decomposition})} \}.$$
Hence $\Gamma$ is naturally endowed with a metric space structure,
with the distance given by
\begin{equation}\label{distance on Ga}
\rm{dist}(\alpha, \beta) = \| \alpha^{-1} \beta \|.
\end{equation}
For $\gamma \in \Gamma$, we denote by $D_n(\gamma)$ the disk of
radius $n$ centered at $\gamma$ and by $D_n$ the disk $D_n(1)$.
Notice that $D_1$ is the set of generators of $\Gamma$. The set
$D_n$ provides, by left translation, a \emph{neighborhood of
$\gamma$}, that is $\gamma D_n = D_n(\gamma)$. Indeed, if $\alpha
\in \gamma D_n$ then $\alpha = \gamma \beta$ with $\| \beta \| \leq
n$. Hence ${\rm dist}(\alpha, \gamma) = \| \alpha^{-1} \gamma \| =
\| \beta^{-1} \| \leq n$. Conversely, if $\alpha \in D_n(\gamma)$
then $\| \gamma^{-1} \alpha \| \leq n$, that is $\gamma^{-1} \alpha
\in D_n$). Hence $\alpha = \gamma \gamma^{-1} \alpha \in \gamma
D_n$.

\medskip

The \emph{Cayley graph} of $\Gamma$, has state set $\Gamma$ and
there is an edge from $\gamma$ to $\bar \gamma$ if there exists a
generator $x$ such that $\gamma x = \bar \gamma$. Hence the distance
defined in (\ref{distance on Ga}) coincides with the graph distance
(that is, the minimal length of a path between to states) on the
Cayley graph of $\Gamma$. Notice that, by the symmetry of the set of
generators, this graph is non--oriented. For example, we may look at
the classical cellular decomposition of Euclidean space
$\mathbb{R}^n$ as the Cayley graph of the group $\mathbb{Z}^n$ with
the presentation $\langle x_1, x_1^{-1}, \dots, x_n, x_n^{-1} \mid
x_i x_j = x_j x_i \rangle$.

\medskip

We recall that the function $g : \mathbb{N} \rightarrow \mathbb{N}$
defined by
$$g(n) = |D_n|$$
which counts the elements of the disk $D_n$, is called \emph{growth
function} of $\Gamma$. One can prove that the limit
$$\lambda = \lim_{n \rightarrow \infty} g(n)^{\frac 1 n}$$
always exists. If $\lambda > 1$ then, for all sufficiently large
$n$, we have $g(n) \geq \lambda^n$, and the group $\Gamma$ has
\emph{exponential growth}. If $\lambda = 1$, we distinguish two
cases. Either there exists a polynomial $p(n)$ such that, for all
sufficiently large $n$, we have $g(n) \leq p(n)$, in which case
$\Gamma$ has \emph{polynomial growth}. Otherwise $\Gamma$ has
\emph{intermediate growth} (i.e. $g(n)$ grows faster than any
polynomial in $n$ and slower then any exponential function $x^n$
with $x > 1$). Moreover, it is possible to prove that the type of
growth is a property of the group $\Gamma$ (i.e. it does not depend
on the choice of a set of generators). For this reason we deal with
the \emph{growth of a group}. This notion has been independently
introduced by Milnor \cite{Milnor68}, Efremovi\v c
\cite{Efremovic53} and \v Svarc \cite{Svarc55} and it is very useful
in the theory of cellular automata.

A group $\Gamma$ is \emph{amenable} if it admits a
$\Gamma$--invariant probability measure, that is a function $\mu :
2^\Gamma \longrightarrow [0 , 1]$ defined on the subsets of $\Gamma$
such that for $A, B \subseteq \Gamma$ and for every $\gamma \in
\Gamma$
\begin {itemize}
\item[--] $A \cap B = \emptyset \Rightarrow \mu(A \cup B) = \mu(A) + \mu(B)$
(\emph{finite additivity})
\item[--] $\mu(\gamma A) = \mu(A)$ (\emph{$\Gamma$--invariance})
\item[--] $\mu(\Gamma) = 1$ (\emph{normalization}).
\end{itemize}
Finite groups, abelian groups, solvable groups, subgroups of
amenable groups are all examples of amenable groups. Moreover, a
finitely generated group of non--exponential growth is amenable. The
free group $\mathbb{F}_2$ of rank 2 has exponential growth and is
non--amenable, but there exist examples of amenable groups of
exponential growth~\cite{CeccheriniGrigorchuk97}.

\section{Shift spaces and cellular automata}\label{Shift spaces and cellular automata}
Let $\mathcal{A}$ be a finite set (with at least two elements)
called \emph{alphabet}. Let $\Gamma$ a finitely generated group as
in the previous section. An element of the set $\mathcal{A}^\Gamma$
(i.e. the set of all functions $c : \Gamma \rightarrow
\mathcal{A}$), is called a \emph{configuration}. If $\Gamma =
\mathbb{Z}$, a configuration is clearly a bi--infinite word over the
alphabet $\mathcal{A}$.

For every $X \subseteq \mathcal{A}^\Gamma$ and $E \subseteq \Gamma$,
we denote by $X_E$ the restrictions to $E$ of the configurations in
$X$, that is
$$X_E = \{ c_{|E} \mid c \in X \}.$$
A \emph{pattern of X} is an element of $X_E$ where $E$ is a
non--empty finite subset of $\Gamma$. The set $E$ is called the
\emph{support} of the pattern. A \emph{block} of X is a pattern of
$X$ whose support is a disk. If $X \subseteq \mathcal{A}^\mathbb{Z}$
(i.e. $\Gamma = \mathbb{Z}$), a block of $X$ (also called a
\emph{factor}) is a finite word appearing in some bi--infinite word
of $X$. In this case, the \emph{language of X} is the set
$\mathcal{L}(X)$ of all factors of $X$.

A subset $X \subseteq \mathcal{A}^\Gamma$ is \emph{subshift} if
there exists a set of blocks $\mathcal{F} \subseteq \bigcup_{n \in
\mathbb{N}} \mathcal{A}^{D_n}$ such that $X =
\mathsf{X}_\mathcal{F}$, where
$$\mathsf{X}_\mathcal{F} = \{ c \in \mathcal{A}^\Gamma \mid {c^{\alpha}}_{|D_n}
\notin \mathcal{F}\ {\rm for\ every} \ \alpha \in \Gamma, n \in
\mathbb{N} \}.$$ In this case, $\mathcal{F}$ is a \emph{set of
forbidden blocks} of $X$. A subshift is also indifferently called
\emph{shift space} or simply \emph{shift}. An
\emph{$n$--dimensional} shift is a subshift of
$\mathcal{A}^{\mathbb{Z}^n}$.

\medskip

On the set $\mathcal{A}^\Gamma$ is defined the usual product
topology, where the topology in $\mathcal{A}$ is the discrete one.
By Tychonoff's theorem, $\mathcal{A}^\Gamma$ is also compact.  If
$c_1, c_2 \in \mathcal{A}^\Gamma$ are two configurations, we define
the distance
$${\rm dist}(c_1,c_2) = \frac 1 {n+1}$$
where $n$ is the least natural number such that $c_1 \neq c_2$ in
$D_n$. If such an $n$ does not exist, that is if $c_1 = c_2$, we set
their distance equal to zero. Notice that the topology induced by
this metric is equivalent to the product topology.

The group $\Gamma$ acts on $\mathcal{A}^\Gamma$ on the right as
follows:
$$(c^{\gamma})_{|\alpha} = c_{|\gamma\alpha}$$
for each $c \in \mathcal{A}^\Gamma$ and each $\gamma, \alpha \in
\Gamma$ (where $c_{|\alpha}$ denotes the value of $c$ at $\alpha$).

In \cite{CeccheriniFiorenziScarabotti04} we prove that the
combinatorial definition of a shift space is equivalent to a
topological one. This fact is well known in the $n$--dimensional
case.

\begin{proposition}{\rm \cite[Proposition 4.3]{{CeccheriniFiorenziScarabotti04}}}\label{shift iff forbidden}
A subset $X \subseteq \mathcal{A}^\Gamma$ is a shift if and only if
it is topologically closed and $\Gamma$--invariant (i.e. $X^{\Gamma}
= X$).
\end{proposition}

\noindent {\bf Remark.} A pattern with support $E$ is a function $p
: E \rightarrow \mathcal{A}$. If $\gamma \in \Gamma$, we have that
the function $\bar p : \gamma E \rightarrow \mathcal{A}$ defined as
$\bar p_{|\gamma \alpha} = p_{|\alpha}$ (for each $\alpha \in E$),
is the pattern obtained copying $p$ on the translated support
$\gamma E$. Moreover, if $X$ is a shift, we have that $\bar p \in
X_{\gamma E}$ if and only if $p \in X_E$. For this reason, in the
sequel we do not make distinction between $p$ and $\bar p$ (when the
context makes it possible). For example, a word $a_1 \dots a_n$ is
simply a finite sequence of symbols for which we do not specify (if
it is not necessary), if the support is the interval $[1,n]$ or the
interval $[2,n+1]$.

\medskip

Let $X$ be a subshift of $\mathcal{A}^\Gamma$. A function $\tau : X
\rightarrow \mathcal{A}^\Gamma$ is \emph{M--local} if there exists
$\delta : X_{D_M} \rightarrow \mathcal{A}$ such that for every $c
\in X$ and $\gamma \in \Gamma$
$$(\tau(c))_{|\gamma} = \delta({c^{\gamma}}_{|D_M}) =
\delta(c_{|\gamma\alpha_1}, c_{|\gamma\alpha_2}, \dots,
c_{|\gamma\alpha_m}),$$
where $D_M = \{ \alpha_1, \dots, \alpha_m \}$. Hence locality means
that the value of $\tau(c)$ at $\gamma$ only depends on the values
of $c$ at the elements of a fixed neighborhood of $\gamma$.

A local function defined on the whole $\mathcal{A}^\Gamma$ is called
a \emph{cellular automaton}.

\medskip

\noindent {\bf Remark.} In the definition of locality, we assume
that the alphabet of the shift $X$ is the same as the alphabet of
its image $\tau(X)$. In this assumption there is no loss of
generality because if $\tau : X \subseteq \mathcal{A}^\Gamma
\rightarrow \mathcal{B}^{\Gamma}$, one can always consider $X$ as a
shift over the alphabet $\mathcal{A} \cup \mathcal{B}$.

\medskip

The following characterization of local functions is, in the
one--dimensional case, known as the Curtis--Lyndon--Hedlund theorem.
In~\cite{CeccheriniFiorenziScarabotti04} it has been generalized as
follows to any local function.

\begin{proposition}{\rm \cite[Proposition 4.4]{{CeccheriniFiorenziScarabotti04}}}\label{local}
A function $\tau : X \rightarrow \mathcal{A}^\Gamma$ is local if and
only if it is continuous and commutes with the $\Gamma$--action
(i.e. for each $c \in X$ and each $\gamma \in \Gamma$, one has
$\tau(c^{\gamma}) = \tau(c)^{\gamma}$).
\end{proposition}
As a consequence of this fact, we have that \emph{the composition of
two local functions is still local}.

\begin{proposition}\label{continuous}
Let $X$ be a shift. For each $\gamma \in \Gamma$ the function $X
\rightarrow X$ that associates with each $c \in X$ its translated
configuration $c^{\gamma}$, is continuous.
\end{proposition}

\begin{proof}
Let $n \geq 0$ and let $m \geq 0$ such that $\gamma D_n \subseteq
D_m$. If ${\rm dist}(c , \bar c) < \frac 1 {m+1}$, then $c$ and
$\bar c$ agree on $D_m$ and therefore on $\gamma D_n$. Hence $\alpha
\in D_n \Rightarrow c_{|\gamma\alpha} = {\bar{c}}_{|\gamma\alpha}
\Rightarrow {c^\gamma}_{|\alpha} = {{{\overline
c}}^\gamma}_{|\alpha}$. That is $c^{\gamma}$ and
${\bar{c}}^{\gamma}$ agree on $D_n$ so that ${\rm dist}(c^{\gamma} ,
{\bar{c}}^{\gamma}) < \frac 1 {n+1}$.
\end{proof}

\medskip

\noindent {\bf Remark.} Notice that, in general, this function does
not commute with the $\Gamma$--action (and therefore it is not
local). Indeed, if $\Gamma$ is not abelian and $\gamma \alpha \neq
\alpha \gamma$, we may have $(c^{\gamma})^{\alpha} \neq
(c^{\alpha})^{\gamma}$.

\medskip

If $X$ is a subshift of $\mathcal{A}^\Gamma$ and $\tau : X
\rightarrow \mathcal{A}^\Gamma$ is a local function, Proposition
\ref{local} guarantees that the image $Y = \tau(X)$ is still a
subshift of $\mathcal{A}^\Gamma$. Indeed $Y$ is closed (or,
equivalently, compact) and it is also $\Gamma$--invariant. In fact
we have that $Y^{\Gamma}$ = $(\tau(X))^{\Gamma}$ =
$\tau(X^{\Gamma})$ = $\tau(X) = Y$. Moreover, if $\tau$ is injective
then $\tau : X \rightarrow Y$ is a homeomorphism and $\tau^{-1}$
commutes with the $\Gamma$--action. Indeed, if $c \in Y$ then $c =
\tau(\bar c)$ for a unique $\bar c \in X$ and we have
$$\tau^{-1}(c^{\gamma}) = \tau^{-1}(\tau(\bar c)^{\gamma}) =
\tau^{-1}(\tau({\bar c}^{\gamma})) = {\bar c}^{\gamma} =
(\tau^{-1}(c))^{\gamma}.$$ By Proposition \ref{local}, we have that
$\tau^{-1}$ is local and the well--known Richardson's theorem
\cite{Richardson72}, stating that the inverse of an invertible
$n$--dimensional cellular automaton is a cellular automaton, holds
also in this more general setting. In the one--dimensional case,
Lind and Marcus \cite[Theorem 1.5.14]{LindMarcus95} give a direct
proof of this fact. This result leads us to say that two subshifts
are \emph{conjugate} if there exists a local bijective function
between them (namely a \emph{conjugacy}). The \emph{invariants} are
the properties of a shift which are invariant under conjugacy.

\subsection{Irreducibility}
A one--dimensional shift $X \subseteq \mathcal{A}^\mathbb{Z}$ is
\emph{irreducible} if for each pair of words $u, v \in
\mathcal{L}(X)$, there exists a word $w \in \mathcal{L}(X)$ such
that the concatenated word $u w v \in \mathcal{L}(X)$. This
corresponds to the transitivity of the related discrete--time
dynamical system. The natural generalization of this property to any
subshift $X \subseteq \mathcal{A}^\Gamma$ is that for each pair of
patterns $p_1 \in X_E$ and $p_2 \in X_F$, there exists an element
$\gamma \in \Gamma$ such that $E \cap \gamma F = \emptyset$ and a
configuration $c \in X$ such that $c_{|E} = p_1$ and $c_{|\gamma F}
= p_2$. In other words, a shift is irreducible if whenever we have
two patterns appearing each one in some configuration of $X$, there
exists a configuration $c \in X$ in which these two patterns appear
simultaneously on disjoint supports. In the one--dimensional case,
these two definitions are equivalent, as proved
in~\cite[Section~2]{Fiorenzi03}.

A stronger notion is that of \emph{mixing} shift: for each pair of
patterns $p_1 \in X_E$ and $p_2 \in X_F$, there exists $M
> 0$ such that for each $\gamma \notin D_M$ there is a configuration
$c \in X$ such that $c_{|E} = p_1$ and $c_{|\gamma F} = p_2$ (notice
that if $M$ is big enough, then $E \cap \gamma F = \emptyset$). In
other words, a shift $X$ is mixing if and only if for each pair of
open sets $U, V \subseteq X$ there exists $M > 0$ such that $U \cap
V^{\gamma} \neq \emptyset$ for all $\gamma \notin D_M$. Indeed,
given a pattern $p$ with support $E$, consider the set $U = \{ c \in
X \mid c_{|E} = p \}$. If $E = \{\gamma_1, \dots, \gamma_n \}$ then
$U = \bigcap_{i = 1}^n \{ c \in X \mid c_{|\gamma_i} = p_{|\gamma_i}
\}$ is a finite intersection of open sets and hence is open.

Further forms of irreducibility has been introduced in
\cite{Fiorenzi03} and \cite{Fiorenzi04}. First, \emph{strong
irreducibility} states that if the supports of the patterns are far
enough, than it is not necessary to translate them in order to find
a configuration in which both the patterns appear. Hereafter,
\emph{semi--strong irreducibility} states that if the supports of
the patterns are far enough, than translating them ``a little'' (the
length of this difference being bounded and only depending on the
shift), we can find a common extension.

\subsection{Shifts of finite type}
A shift is \emph{of finite type} if it admits a finite set of
forbidden blocks. Hence we can decide whether or not a configuration
belongs to such a shift only checking its blocks of a fixed (and
only depending on the shift) radius.

More precisely, if $X$ is a shift of finite type, since a finite set
$\mathcal{F}$ of forbidden blocks of $X$ has a maximal support, we
can always suppose that each block of $\mathcal{F}$ has the disk
$D_M$ as support (indeed each block that contains a forbidden block
is forbidden). In this case the shift $X$ is called \emph{M--step}
and the number $M$ is called the \emph{memory of X}.

If $X$ is a one--dimensional shift, we define the memory of $X$ as
the number $M$, where $M+1$ is the maximal length of a forbidden
word. For these shifts we have the following useful property:
\begin{proposition}{\rm \cite[Theorem 2.1.8]{LindMarcus95}}\label{pasting in Z}
A shift $X \subseteq \mathcal{A}^\mathbb{Z}$ is an $M$--step shift
of finite type if and only if whenever $uv, vw \in \mathcal{L}(X)$
and $|v| \geq M$, then $uvw \in \mathcal{L}(X)$ (where $|v|$ denotes
the length of the word $v$).
\end{proposition}
As proved in~\cite[Corollary 2.11]{Fiorenzi03}, this ``overlapping''
property holds more generally for subshifts of finite type of
$\mathcal{A}^\Gamma$.

\section{One--dimensional shifts spaces}\label{One--dimensional shifts spaces}
\subsection{Edge shifts}\label{Edge shifts} A relevant class of
one--dimensional subshifts of finite type, is that of edge shifts.
This class is strictly tied up with that of finite graphs. This
relation allows us to study the properties of an edge shift
(possible quite complex) by studying the properties of its underling
graph.

More precisely, let $\mathsf{G} = (Q, \mathcal{E})$ be a finite
directed multigraph with state set $Q$ and edge set $\mathcal{E}$.
The \emph{edge shift} $X_\mathsf{G}$ is the subshift of
$\mathcal{E}^\mathbb{Z}$ defined by
$$X_{\mathsf{G}} = \{ (e_z)_{z
\in \mathbb{Z}} \in \mathcal{E}^\mathbb{Z} \mid \mathsf{t}(e_z) =
\mathsf{i}(e_{z+1})\ {\rm for\ all\ }z \in \mathbb{Z} \},$$ where
the edge $e \in \mathcal{E}$ has initial state $\mathsf{i}(e)$ and
terminal state $\mathsf{t}(e)$.

%The structure of a finite graph (and hence of an edge shift) can be
%easily described by a matrix, the so--called \emph{adjacency matrix
%of} $\mathbf{G}$.
%
%More precisely, let $\mathbf{G}$ be a graph with vertex set
%$\mathcal{V} = \{ 1, \dots, n \}$; the \emph{adjacency matrix of}
%$\mathbf{G}$ is the matrix {\bf A} such that ${\mathbf A}_{ij}$ is
%the number of edges in $\mathbf{G}$ with initial state $i$ and
%terminal state $j$.
%
%It is easily seen that the $(i,j)$--entry of the matrix ${\mathbf
%A}^m$ (the product of {\bf A} with itself $m$ times), is the number
%of paths in $\mathbf{G}$ with length $m$ from $i$ to $j$.
%
%The edge shift $X_{\mathbf G}$ is sometimes denoted as $X_{\mathbf
%A}$, where {\bf A} is the adjacency matrix of $\mathbf{G}$. The
%fundamental role of the adjacency matrix will be clarified in the
%sequel.

Each one--dimensional shift of finite type is conjugate to an edge
shift and hence they have the same invariants. Thus, also in this
case, the properties of the shift depend on the structure of a
suitable graph. For this, it is easy to see that every edge shift is
a 1--step shift of finite type with set of forbidden blocks $\{ ef
\mid e, f \in \mathcal{E} \ {\rm and}\ \mathsf{t}(e) \neq
\mathsf{i}(f) \}$. Conversely, given a $M$--step shift of finite
type $X$, we give an effective procedure to construct a suitable
graph $\mathsf{G}$ such that $X$ is conjugate to $X_\mathsf{G}$. The
states of $\mathsf{G}$ are the words of $\mathcal{L}(X)$ of length
$M$ and there is an edge from state $a_1 \dots a_M$ to state $a_2
\dots a_{M+1}$ if the word $a_1 \dots a_M a_{M+1}$ still belongs to
$\mathcal{L}(X)$.

\begin{figure}[htbp]
\begin{center}
\begin{picture}(50,6)
\gasset{Nframe=n,Nadjust=w,Nh=6,Nmr=0} \node(1)(0,0){$a_1 \dots
a_M$} \node(2)(50,0){$a_2 \dots a_{M+1}$}
\drawedge[curvedepth=0](1,2){}
\end{picture}
\end{center}
\begin{center}
\bigskip if $a_1 \dots a_M a_{M+1} \in \mathcal{L}(X)$
\end{center}
\end{figure}

The edge shift accepted by this graph is the $(M+1)$th \emph{higher
block shift} of $X$ and is denoted by $X^{[M+1]}$. The shifts $X$
and $X^{[M+1]}$ are conjugate by the function $\tau : X^{[M+1]}
\rightarrow X$ defined by setting $\tau(c)_{|z}$ equal to the first
letter of the word $c_{|z}$, for each $c \in X^{[M+1]}$ and each $z
\in \mathbb{Z}$. This function is bijective and local. The table
below points out its behavior.

\begin{figure}[htbp]
\begin{center}
\begin{tabular}{|c|c|c|c|c|}
\hline
\dots & $a_{-1} a_0 \dots a_M$ & $a_0 a_1 \dots a_{M+1}$ & $a_1 a_2 \dots a_{M+2}$ & \dots \\
\hline
\dots & $a_{-1}$ & $a_0$ & $a_1$ & \dots \\
\hline
\end{tabular}
\end{center}
\end{figure}

Notice that in a graph $\mathsf{G}$, there can be a state from which
no edges start or at which no edges end. Such a state is called
\emph{stranded}. Clearly no bi--infinite paths in $X_{\mathsf G}$
involve a stranded state, hence the stranded states and the edges
starting or ending at them are inessential for the edge shift
$X_{\mathsf G}$. Following Lind and Marcus \cite[Definition
2.2.9]{LindMarcus95}, a graph is \emph{essential} if no state is
stranded. Removing step by step the stranded states of $\mathsf{G}$,
we get an essential graph $\bar{\mathsf G}$ that recognizes the same
edge shift. This procedure is effective, because $\mathsf{G}$ has a
finite number of states. Moreover, this ``essential form'' of
$\mathsf{G}$ is unique.

\subsection{Sofic shifts}\label{Sofic shifts}
The class of \emph{sofic} shifts has been introduced by Weiss in
\cite{Weiss73} as the smallest class of shifts containing the shifts
of finite type and closed under \emph{factorization} (i.e. the image
under a local map). Equivalently, one can see that a sofic shift is
the set of labels of bi--infinite paths in a finite automaton.

More precisely, a \emph{finite automaton} $\mathsf{A}$ is a finite
directed multigraph labeled by a finite alphabet. A subshift of
$\mathcal{A}^\mathbb{Z}$ is sofic if and only if it is the set of
the labels of all the bi--infinite paths on a finite automaton
$\mathsf{A}$ labeled by $\mathcal{A}$. In this case we say that the
shift is \emph{accepted} by $\mathsf{A}$ and it is denoted
$X_\mathsf{A}$. The automaton $\mathsf{A}$ is called
\emph{presentation} of the shift.

Obviously each edge shift is sofic, where the label of an edge is
the edge itself. Hence, by conjugation, each shift of finite type is
sofic. An accepting automaton is given below by considering the
graph $\mathsf{G}$ with labeling $\tau$ introduced in the previous
section.

\begin{figure}[htbp]
\begin{center}
\begin{picture}(50,8)
\gasset{Nframe=n,Nadjust=w,Nh=6,Nmr=0} \node(1)(-25,0){$a_1 \dots
a_M$} \node(2)(25,0){$a_2 \dots a_{M+1}$} \node(3)(75,0){$a_3 \dots
a_{M+2}$} \drawedge[curvedepth=0](1,2){$a_1$}
\drawedge[curvedepth=0](2,3){$a_2$}
\end{picture}
\end{center}
\begin{center}
\end{center}
\end{figure}

\noindent {\bf Remarks.} (1) If we deal only with essential graphs,
the language of a sofic shift is regular. (2) A sofic shift is
irreducible if and only if it has a strongly connected presentation.

\medskip

An automaton is \emph{deterministic} if for any state and any
symbol, there is at most one outgoing edge labeled by this symbol.
An irreducible sofic shift has a unique (up to isomorphisms of
automata) \emph{minimal deterministic presentation}, that is a
deterministic presentation having the fewest states among all
deterministic presentations of the shift \cite[Theorem
3.3.2]{LindMarcus95}. Lind and Marcus also proved in \cite[Lemma
3.3.10]{LindMarcus95} that the minimal deterministic presentation of
an irreducible sofic shift is strongly connected and, in
\cite[Proposition 2.2.14]{LindMarcus95}, that if $\mathsf{G}$ is a
strongly connected graph, then the edge shift $X_\mathsf{G}$ is
irreducible. As a consequence of this two facts, we have the
following corollary.

\begin{corollary}\label{irreducible sofic}
A subset $X \subseteq \mathcal{A}^\mathbb{Z}$ is an irreducible
sofic shift if and only if it is the image under a local function of
an irreducible shift of finite type.
\end{corollary}

\begin{proof}
Let $X$ be an irreducible sofic shift and let $\mathsf{G}$ be the
underlying graph of the minimal deterministic presentation of $X$.
Then the edge shift $X_\mathsf{G}$ is irreducible. Conversely, the
image under a local function of an irreducible shift is also
irreducible.
\end{proof}

%Notice that the strong connectedness of $\mathsf{G}$ is equivalent
%to the following property of the $n \times n$ adjacency matrix {\bf
%A} of $\mathsf{G}$: for each $i, j \in \{ 1, \dots, n \}$, there
%exists an $m \in \mathbb{N}$ such that the $(i,j)$--entry of
%${\mathbf A}^m$ is not zero.

\subsection{Decision problems}\label{Decision problems}

A natural decision problem arising in the theory of cellular
automata concerns the existence of effective procedures to establish
the surjectivity and the injectivity of the transition function.
Amoroso and Patt have shown in \cite{AmorosoPatt72} that there are
algorithms to decide surjectivity and injectivity of
one--dimensional cellular automata. On the other hand Kari has shown
in \cite{Kari90} and \cite{Kari94} that both the injectivity and the
surjectivity problems are undecidable for $n$--dimensional cellular
automata with $n > 1$.

In this section we extend the problem to local function over
subshifts of finite type of $\mathcal{A}^\mathbb{Z}$, giving in both
cases a positive answer to the existence of decision procedures.
More decision problems will be stated in Section~\ref{Further
decision problems}.

\subsubsection{A decision procedure for surjectivity}
If $X$ is a shift of finite type and $Y$ is sofic, the problem of
deciding whether or not a function $\tau : X \rightarrow Y$ given in
terms of local map is surjective, is decidable.

Let $\tau : X \rightarrow \mathcal{A}^\mathbb{Z}$ be a function
defined by a local rule $\delta$. Without loss of generality, we can
assume that $X$ has memory $2M$ and that $\tau$ is $M$--local. The
function $\tau$ can be represented in this way: consider the
presentation of the edge shift $X^{[2M+1]}$ constructed in
Section~\ref{Edge shifts}. The label of the edge between $u_1 \dots
u_M\ a\ v_1 \dots v_{M-1}$ and $u_2 \dots u_M\ a\ v_1 \dots v_{M}$
(both in $\mathcal{L}(X)$), is the letter $\delta(u_1 \dots u_M\ a\
v_1 \dots v_M)$, that is the letter to write in place of $a$ in the
image block:

\begin{figure}[htbp]
\begin{center}
\begin{picture}(80,8)
\gasset{Nframe=n,Nadjust=w,Nh=6,Nmr=0} \node(1)(0,0){$u_1 \dots u_M\
a\ v_1 \dots v_{M-1}$} \node(2)(80,0){$u_2 \dots u_M\ a\ v_1 \dots
v_M$} \drawedge[curvedepth=0](1,2){$\delta(u_1 \dots u_M\ a\ v_1
\dots v_M)$}
\end{picture}
\end{center}
\begin{center}
\end{center}
\end{figure}

\noindent In this way we get a finite automaton $\mathsf{A}$ which
is the presentation of the (sofic) shift $\tau(X)$. To see whether
or not the function $\tau$ is surjective, Lind and Marcus give in
\cite[Section 3.4]{LindMarcus95} an effective procedure to decide
whether two finite automata accept the same shifts.

\subsubsection{A decision procedure for injectivity}

If $X$ is a shift of finite type, the problem of deciding whether or
not a function $\tau : X \rightarrow \mathcal{A}^\mathbb{Z}$ given
in terms of local map is injective, is decidable.

As we have seen above, we can construct a finite automaton
$\mathsf{A}$ which is a presentation of the sofic shift $\tau(X)$.
From $\mathsf{A}$, we construct another automaton
$\mathsf{A}*\mathsf{A}$. Its states are couples $(p, q)$, where $p$
and $q$ are states of $\mathsf{A}$. There is an edge $(p, q)
\stackrel{a}{\longrightarrow} (r, s)$ labeled $a$, if in
$\mathsf{A}$ there are two edges labeled $a$ such that:
$$p \stackrel{a}{\longrightarrow} r \quad {\rm and} \quad q \stackrel{a}{\longrightarrow} s.$$
Notice that, in general, $X_\mathsf{A} = X_{\mathsf{A}*\mathsf{A}}$
and hence $\mathsf{A}*\mathsf{A}$ is another presentation of
$\tau(X)$.

A state $(p, q)$ of $\mathsf{A}*\mathsf{A}$ is \emph{diagonal} if $p
= q$. Notice that the function $\tau$ is non--injective if and only
if on the graph $\mathsf{A}$ there are two different bi--infinite
paths with the same label. This fact is equivalent to the existence
of a bi--infinite path on $\mathsf{A}*\mathsf{A}$ that involves a
non--diagonal state. Hence, starting from the graph
$\mathsf{A}*\mathsf{A}$ we construct an essential graph that accepts
the same bi--infinite paths. It suffices to check, on this latter
graph, if some non--diagonal state is left.

\section{Density of periodic configurations}\label{Density of periodic
configurations} In the $n$--dimensional case, a periodic
configuration is obtained ``repeating'' in each direction the same
finite block. Hence, translating such a configuration, we get only a
finite number of new configurations. This property leads us to
define \emph{periodic} a configuration whose $\Gamma$--orbit is
finite.

In this section we establish some generalities about periodic
configurations. We also prove that the density of periodic
configurations is an invariant of the shifts, as is the number of
the periodic configuration with a fixed period.

\begin{definition}
{\rm A configuration $c \in \mathcal{A}^\Gamma$ is
\emph{n--periodic} if its orbit $c^{\Gamma} = \{ c^{\gamma} \mid
\gamma \in \Gamma \}$ consists of $n$ elements. In this case $n$ is
the \emph{period} of $c$. A configuration is \emph{periodic} if it
is $n$--periodic for some $n \in \mathbb{N}$.}
\end{definition}

From now on, $\mathcal{Q}_n$ (resp. $\mathcal{P}_n$) denotes the set
of the periodic configurations with period (resp. dividing) $n$ and
$\mathcal{P}$ is the set $\bigcup_{n \geq 1}\mathcal{P}_n =
\bigcup_{n \geq 1}\mathcal{Q}_n$ of all periodic configurations.

\medskip

In general, a configuration $c \in \mathcal{A}^\Gamma$ is constant
on the right cosets of its own stabilizer $H_c$ (i.e. the subgroup
of all $\gamma \in \Gamma$ such that $c^\gamma = c$). Indeed, if
$\gamma \in H_c$, we have
$$c_{|\gamma \alpha} = (c^\gamma)_{|\alpha} = c_{|\alpha}.$$
Hence, if $c$ is periodic it is constant on the right cosets of a
subgroup of finite index. Now we prove that this property
characterizes periodic configurations.

\begin{proposition}\label{periodic iff constant}
A configuration $c \in \mathcal{A}^\Gamma$ belongs to
$\mathcal{P}_n$ if and only if there exists a subgroup $H \leq
\Gamma$ with finite index dividing n, such that c is constant on the
right cosets of H.
\end{proposition}

\begin{proof}
Let $c$ be $m$--periodic with $m | n$. By definition, the stabilizer
$H_c$ has finite index $m$ and, as we seen, $c$ is constant on the
right cosets of $H_c$. Conversely, if $H$ has finite index dividing
$n$ and $c$ is constant on the right cosets of $H$, we have that $H
\subseteq H_c$. Indeed, if $\gamma \in H$ and $\alpha \in \Gamma$ we
have $(c^\gamma)_{|\alpha} = c_{|\gamma \alpha} = c_{|\alpha}$ and
hence $c^\gamma = c$ so that $\gamma \in H_c$. Since $H$ is of
finite index, $H_c$ has finite index as well and the index of $H_c$
divides that of $H$ so that it divides $n$. Hence $c \in
\mathcal{P}_n$.
\end{proof}

\begin{corollary}\label{Pn finito}
The set $\mathcal{P}_n$ is finite.
\end{corollary}

\begin{proof}
By Proposition \ref{periodic iff constant}, a configuration $c \in
\mathcal{A}^\Gamma$ belongs to $\mathcal{P}_n$ if and only if it is
constant on the right cosets of a subgroup $H$ with finite index
dividing $n$. Being $\Gamma$ a finitely generated group, there are
finitely many subgroups of $\Gamma$ of a fixed finite index (see,
for example, \cite{Rotman95}). Thus, these subgroups $H$ are in
finite number. For a fixed $H$ among them, there are finitely many
functions from the right cosets of $H$ to $\mathcal{A}$, that is
$|\mathcal{A}|^{[\Gamma : H]}$.
\end{proof}

For a shift space $X$, we denote by $\mathcal{P}(X)$ the set of all
periodic configuration of $X$, that is $\mathcal{P}(X) = \mathcal{P}
\cap X$. Similarly, we denote by $\mathcal{Q}_n(X)$ (resp.
$\mathcal{P}_n(X)$), the intersection $\mathcal{Q}_n \cap X$ (resp.
$\mathcal{P}_n \cap X$), and $q_n(X)$ (resp. $p_n(X)$), denotes its
cardinality.

\begin{proposition}\label{dense}
If $X \subseteq \mathcal{A}^\Gamma$ is a shift such that
$\mathcal{P}(X)$ is dense in X and $\tau : X \rightarrow
\mathcal{A}^\Gamma$ is a local function, then $\mathcal{P}(\tau(X))$
is dense in $\tau(X)$.
\end{proposition}

\begin{proof}
Set $Y = \tau(X)$. First we prove that $\tau(\mathcal{P}(X))
\subseteq \mathcal{P}(Y)$. Indeed if $c \in X$ the stabilizer $H_c$
is contained in $H_{\tau(c)}$ (we have $c^\gamma = c \Rightarrow
(\tau(c))^\gamma = \tau(c^\gamma) = \tau(c)$), and if $H_c$ has
finite index, then $H_{\tau(c)}$ has finite index as well. Then
$\overline{\mathcal{P}(Y)} \supseteq \overline{\tau(\mathcal{P}(X))}
\supseteq \tau(\overline{\mathcal{P}(X)}) = \tau(X) = Y$.
\end{proof}

\begin{corollary}
The density of its periodic configurations is an invariant of the
shift.
\end{corollary}

\begin{proposition}
Let $X \subseteq \mathcal{A}^\Gamma$, then the numbers $q_n(X)$ and
$p_n(X)$ are invariants of $X$.
\end{proposition}

\begin{proof}
Let $\tau : X \rightarrow Y$ be a conjugacy. We prove that $H_c =
H_{\tau(c)}$. As proved in Proposition \ref{dense}, we always have
that $H_c \subseteq H_{\tau(c)}$. Conversely, if $\gamma \in
H_{\tau(c)}$, we have $\tau(c^\gamma) = \tau(c)^\gamma = \tau(c)$.
The function $\tau$ being injective, we have $c^\gamma = c$. Thus,
$c \in \mathcal{Q}_n(X)$ if and only if $\tau(c) \in
\mathcal{Q}_n(Y)$. The same holds for configurations whose period
divides $n$.
\end{proof}

\section{Residually finite groups}\label{Residually finite groups}

A group $\Gamma$ is \emph{residually finite} if for every $\gamma
\in \Gamma \setminus \{1\}$, there exists a subgroup of finite index
$H \leq \Gamma$ such that $\gamma \notin H$. In other words, a group
if residually finite if the intersection of all its subgroup of
finite index is trivial. Examples of residually finite groups are
the groups $\mathbb{Z}^n$ and, in general, all finitely generated
abelian groups. The free group $\mathbb{F}_n$ of rank $n$ is an
example of residually finite, non--abelian group. The additive group
of rational numbers $\mathbb{Q}$ is an example of abelian,
non--finitely generated and non--residually finite group.

In this section we see that (finitely generated) residually finite
groups are precisely those groups such that for each finite set
$\mathcal{A}$, the set $\mathcal{P}$ of periodic configurations is
dense in $\mathcal{A}^\Gamma$.

\medskip

The first part of the following lemma is well known. Its extension
is due to T. Ceccherini--Silberstein and A. Mach\` \i.

\begin{lemma}\label{set}
Let $\Gamma$ be a residually finite group and let $F = \{ \gamma_1,
\dots, \gamma_n \}$ be a finite subset of $\Gamma \setminus \{1\}$.
There exists a subgroup $H \leq \Gamma$ of finite index such that $F
\subseteq \Gamma \setminus H$ and $H \gamma_i \neq H \gamma_j$ for
each $i \neq j$.
\end{lemma}

\begin{proof}
For every $i = 1, \dots, n$ let $H_i$ be a subgroup of finite index
such that
 $\gamma_i \notin H_i$
and let $H_{ij}$ be a subgroup of finite index such that
$\gamma_i\gamma_j^{-1} \notin H_{ij}$ (where $i \neq j$). The
intersection $H$ of all these subgroups has finite index as well.
Moreover $\gamma_i \notin H$ (for each $i$) and
$\gamma_i\gamma_j^{-1} \notin H$ ($i \neq j$).
\end{proof}

\medskip

\noindent {\bf Remark.} In particular, the set $F$ in previous lemma
can be extended to a set of right coset representatives of the
subgroup $H$.

\begin{theorem}\label{residually only if dense}
Let $\Gamma$ be a finitely generated group and $\mathcal{A}$ a
finite alphabet. If $\Gamma$ is residually finite, then the set
$\mathcal{P}$ of periodic configurations is dense in
$\mathcal{A}^\Gamma$.
\end{theorem}

\begin{proof}
Suppose that $\Gamma$ is residually finite. We have to prove that
$\mathcal{A}^\Gamma = \overline{\mathcal{P}}$. Fix $c \in
\mathcal{A}^\Gamma$ and let $H$ be the subgroup of finite index
whose existence is guaranteed by Lemma~\ref{set} with $F = D_n
\backslash \{ 1 \}$, and let $D$ be a set of right coset
representatives of $H$ containing $D_n$. If $\gamma \in \Gamma$ and
$\gamma = h d$ with $h \in H$ and $d \in D$, define a configuration
$c_n$ such that $(c_n)_{|\gamma} = c_{|d}$. This configuration being
constant on the right cosets of $H$, it is periodic. Moreover $c$
and $c_n$ agree on $D_n$ and hence ${\rm dist}(c,c_n) < \frac 1
{n+1}$. Then the sequence of periodic configurations $(c_n)_n$
converges to $c$.\end{proof}

The same result is also given by Yukita \cite{Yukita99}. The
converse of this theorem also holds.

\begin{theorem}\label{residually iff dense}
Let $\Gamma$ be a finitely generated group and $\mathcal{A}$ a
finite alphabet. Then $\Gamma$ is residually finite if and only if
the set $\mathcal{P}$ of periodic configurations is dense in
$\mathcal{A}^\Gamma$.
\end{theorem}

\begin{proof}
If $\Gamma$ is not residually finite, let $\gamma \neq 1$ be an
element belonging to all the subgroups of $\Gamma$ of finite index.
In particular $\gamma \in \bigcap_{c \in \mathcal{P}}H_c$ so that,
for each $c \in \mathcal{P}$, we have $c^{\gamma} = c$ and hence
$c_{|\gamma} = c_{|1}$. Let $\bar c \in \mathcal{A}^\Gamma$ such
that $\bar c_{|\gamma} \neq \bar c_{|1}$, then for each $n$ such
that $\gamma \in D_n$ and each $c \in \mathcal{P}$ we have $\bar
c_{|D_n} \neq c_{|D_n}$. Hence ${\rm dist}(\bar c , c) \geq \frac 1
{n+1}$ and $\bar c \notin {\overline{\mathcal{P}}}$.
\end{proof}

\section{Group shifts}\label{Group shifts}
If the alphabet $\mathcal{A}$ is a finite group, the full shift
$\mathcal{A}^{\Gamma}$ is also a group with product defined as in
the direct product of infinitely  many copies of $\mathcal{A}$.
Endowed with this operation the space $\mathcal{A}^{\Gamma}$ is a
compact metric topological group. A subshift $X \subseteq
\mathcal{A}^{\Gamma}$ which is also a subgroup is called \emph{group
shift}.

In this section we prove (as a consequence of a more general theorem
in \cite{KitchensSchmidt89}), that for this class of shifts the
periodic configurations are dense. Moreover, as we will see in
Section~\ref{Further decision problems}, some well--known decision
problems can be solved for the class of group shifts.

Clearly a group shift is also a compact (metric) group. Hence it is
an example of dynamical system $(X, \Gamma)$, where $X$ is a compact
group and $\Gamma$ is a subgroup of the group Aut($X$) of the
automorphisms of $X$ which are also continuous. Indeed the action of
$\Gamma$ defines a subgroup of Aut($X$): for a fixed $\gamma \in
\Gamma$, the bijective function $c \mapsto c^{\gamma}$ from $X$ to
$X$ is obviously a group homomorphism and, as proved in
Lemma~\ref{continuous}, it is also continuous.

If $(X, \Gamma)$ is such a dynamical system, the group $\Gamma$
\emph{acts expansively} on $X$ if there exists a neighborhood $U$ of
the identity $1$ in $X$ such that $\bigcap_{\gamma \in \Gamma}
\gamma(U) = \{ 1 \}$. The set of \emph{$\Gamma$--periodic} points is
the set of points $x \in X$ such that $\{ \gamma(x) \mid \gamma \in
\Gamma \}$ is finite. Clearly it coincides with the set
$\mathcal{P}(X)$ if $X$ is a group shift.

If $X$ is metrizable and $\Gamma$ is an infinite and finitely
generated abelian group, Kitchens and Schmidt prove in \cite[Theorem
3.2]{KitchensSchmidt89} that if $\Gamma$ acts expansively on $X$
then $(X,\Gamma)$ satisfies the descending chain condition (i.e.
each nested decreasing sequence of closed $\Gamma$--invariant
subgroups is finite), if and only if $(X,\Gamma)$ is conjugate to a
dynamical system $(Y,\Gamma)$, where $Y$ is a group subshift of
$\mathcal{A}^{\Gamma}$ and $\mathcal{A}$ is a compact Lie group.
Notice that, in this context, a conjugation is a continuous groups
isomorphism that commutes with the $\Gamma$--action.

A consequence of this fact is the following theorem.

\begin{theorem}{\rm \cite[Corollary 7.4]{KitchensSchmidt89}}
Let X be a compact group and $\Gamma \leq {\rm Aut}(X)$ a finitely
generated, abelian group. If $\Gamma$ acts expansively on X then the
set of $\Gamma$--periodic points is dense in X.
\end{theorem}

Hence we can prove the following result for group shifts.

\begin{corollary}\label{dense in group shift}
Let $\mathcal{A}$ be a finite group and let $\Gamma$ be a finitely
generated, abelian group. If $X \leq \mathcal{A}^\Gamma$ is a group
shift, then the set $\mathcal{P}(X)$ of periodic configurations of X
is dense in X.
\end{corollary}

\begin{proof}
We prove that the group $\Gamma$ acts expansively on $X$. Indeed the
identity in $X$ is the configuration $c$ assuming the constant value
$1_\mathcal{A}$, where $1_\mathcal{A}$ is the identity of
$\mathcal{A}$. Consider the neighborhood $U$ of $c$ consisting of
all those configurations of $X$ assuming the value $1_\mathcal{A}$
at $1_{\Gamma}$. Obviously $\bigcap_{\gamma \in \Gamma} \{
c^{\gamma} \mid c \in U \}
 = \bigcap_{\gamma \in \Gamma} \{ c \in X \mid c_{|\gamma} = 1_\mathcal{A} \}
 = \{ c \}$. \end{proof}

\medskip

In \cite{KitchensSchmidt89} is also proved that if $X$ is a group
shift, then $X$ is of finite type. Indeed the following theorem is
proved.

\begin{theorem}{\rm \cite[Corollary 3.9]{KitchensSchmidt89}}
Let $\mathcal{A}$ be a compact Lie group. If $X \leq
\mathcal{A}^\Gamma$ is a closed $\Gamma$--invariant subgroup there
exists a finite set $D \subseteq \Gamma$ such that
$$X = \{ c \in \mathcal{A}^{\Gamma} \mid {c^{\gamma}}_{|D} \in H\ {\rm for\ every\ } \gamma \in \Gamma \},$$
where $H$ is a closed subgroup of $\mathcal{A}^D$.
\end{theorem}

Hence if $\mathcal{A}$ is finite and $X$ is a group shift, the set
$\mathcal{A}^D \setminus H$ is finite and is a set of forbidden
blocks for $X$. Although this fact, $X$ is not necessarily
irreducible. For example, consider the group shift $\{ \bar 0 , \bar
1 , \overline{01} , \overline{10} \}$ in $(\mathbb{Z} / 2
\mathbb{Z})^{\mathbb{Z}}$ (where, for each finite word $w$, we
denote by $\bar w$ the bi--infinite word $\dots w w w \dots$).

\medskip

Notice that an abelian, finitely generated group $\Gamma$ is also
residually finite. We have another proof of this fact by fixing a
finite group $\mathcal{A}$ and applying Corollary~\ref{dense in
group shift} to the full group shift $\mathcal{A}^\Gamma$. By
Theorem~\ref{residually iff dense} we have that $\Gamma$ is
residually finite.

\section{Surjunctivity}\label{Surjunctivity}
A selfmapping $\tau : X \rightarrow X$ on a set $X$ is
\emph{surjunctive} if it is either non--injective or surjective. In
other words a function is surjunctive if it is not a strict
embedding and hence the implication injective $\Rightarrow$
surjective holds. This notion is due to
Gottschalk~\cite{Gottschalk73}.

The simplest example is that of a \emph{finite} set $X$ and a
selfmapping $\tau : X \rightarrow X$. Clearly each function of this
kind is surjunctive. Another example is given by an endomorphism of
a finite--dimensional vector space and by a regular selfmapping of a
complex algebraic variety (see \cite{Ax68}). Many others examples of
surjunctive functions are given by Gromov in~\cite{Gromov99}.
Moreover, Richardson proves in \cite{Richardson72} that each
$n$--dimensional cellular automaton is surjunctive.

In this section, we consider the surjunctivity of a general cellular
automaton over a group $\Gamma$ and that of a local function on a
subshift. In fact, we prove that if the periodic configurations of a
subshift $X \subseteq \mathcal{A}^\Gamma$ are dense, then a local
function on $X$ is surjunctive.

The following is a sufficient condition for a selfmapping of a
topological space to be surjunctive. Similar conditions are stated
in~\cite{Gromov99}.

\begin{lemma}\label{density}
Let X be a topological space, let $\tau : X \rightarrow X$ be a
closed function and let $(X_i)_{i \in I}$ be a family of subsets of
$X$ such that
\begin{itemize}
\item[--] $X = \overline{\bigcup_{i \in I}{X_i}}$
\item[--] $\tau(X_i) \subseteq X_i$
\item[--] $\tau_{|X_i} : X_i \rightarrow X_i$ is surjunctive
\end{itemize}
then $\tau$ is surjunctive.
\end{lemma}

\begin{proof}
If $\tau$ is injective then, for every $i \in I$, the restriction
$\tau_{|X_i}$ is injective as well. By the hypotheses we have
$\tau(X_i) = X_i$ and hence $\bigcup_{i \in I}X_i = \bigcup_{i \in
I}\tau(X_i) = \tau(\bigcup_{i \in I}X_i) \subseteq
\tau(\overline{\bigcup_{i \in I}X_i}) = \tau(X)$. Then $X =
\overline{\bigcup_{i \in I}X_i} \subseteq \overline{\tau(X)}$, and
$\tau$ being closed we have $X \subseteq \tau(X)$.
\end{proof}

In the following theorem we prove that the density of the periodic
configuration is a sufficient condition for the surjunctivity of a
local function defined on a subshift. By Theorem \ref{residually iff
dense}, we have that from the residual finiteness of $\Gamma$ it
follows that a cellular automaton $\tau : \mathcal{A}^\Gamma
\rightarrow \mathcal{A}^\Gamma$ is surjunctive. The groups
$\mathbb{Z}^n$ being residually finite, we have that this result
generalizes Richardson's theorem.

\begin{theorem}\label{surjunctive}
Let $X \subseteq \mathcal{A}^\Gamma$ be a shift whose set
$\mathcal{P}(X)$ of periodic configurations is dense in $X$. Then
every local function $\tau : X \rightarrow X$ is surjunctive.
\end {theorem}

\begin{proof}
By Corollary \ref{Pn finito} we have that the set $\mathcal{P}_n(X)
= \mathcal{P}_n \cap X$ is finite. As proved for
Proposition~\ref{dense}, we have that if $\tau$ is local then
$\tau(\mathcal{P}_n(X)) \subseteq \mathcal{P}_n(X)$. Hence $\tau$ is
surjunctive by Lemma~\ref{density}.
\end{proof}

\begin{corollary}
If $\Gamma$ is a residually finite group and $\tau :
\mathcal{A}^\Gamma \rightarrow \mathcal{A}^\Gamma$ is a cellular
automaton, then $\tau$ is surjunctive.
\end{corollary}

\noindent {\bf Remark.} The implication injective $\Rightarrow$
surjective in Theorem~\ref{surjunctive} is not invertible. An
example is the following: let $\mathcal{A} = \{ 0,1 \}$ and $\Gamma
= \mathbb{Z}$. Let $\tau$ be the cellular automaton given by the
local rule $\delta : \mathcal{A}^3 \rightarrow \mathcal{A}$ such
that
$$\delta(a_1, a_2, a_3) = a_1 + a_3 \mod 2.$$
The function $\tau$ is surjective and not injective. Indeed if
$(a_z)_{z \in \mathbb{Z}}$ is a configuration in
$\mathcal{A}^\mathbb{Z}$, a pre--image is given by:
$$
\left\{
\begin{array}{l}
b_0 = 0\\
b_1 = 0\\
b_{n+1} = a_n - b_{n-2} \mod 2\ {\rm if}\ n \geq 2\\
b_{-n} = a_{-n+1} - b_{-n+2} \mod 2\ {\rm if}\ n \leq 0\\
\end{array}
\right.
$$
that is

\begin{figure}[htbp]
\begin{center}
\begin{tabular}{|c|c|c|c|c|c|c|c|c|c|c|}
\hline
\dots& $a_{-2} - a_0$ & $a_{-1}$ & $a_0$ & 0 & 0 & $a_1$ & $a_2$ & $a_3 - a_1$ & $a_4 - a_2$ &\dots \\
\hline
\dots& $a_{-3}$ & $a_{-2}$ & $a_{-1}$ & $a_0$ & $a_1$ & $a_2$ & $a_3$ & $a_4$ & $a_5$ &\dots \\
\hline
\end{tabular}\ .
\end{center}
\end{figure}

\noindent By taking $b_0 = 1 = b_1$ we get a different pre--image.

\subsection{Garden of Eden theorem} For $n$--dimensional cellular automata,
Moore~\cite{Moore62} has given a sufficient condition for the
existence of a a pattern without pre--image. Moore's condition (that
is the existence of two different patterns - called \emph{mutually
erasable} - for which each pair of extending configurations that
coincide outside their supports, have the same image) was also
proved to be necessary by Myhill~\cite{Myhill63}. This equivalence
between ``local surjectivity'' and ``local injectivity'' of a
cellular automaton is the classical well--known \emph{Garden of Eden
(GOE) theorem}.

The GOE theorem has been generalized by Machì and
Mignosi~\cite{MachiMignosi93} to any cellular automata defined on a
group of non--exponential growth. Later it has been proved by
Ceccherini \emph{et al.}~\cite{CeccheriniMachiScarabotti99} for the
wider class amenable groups.

It can be proved (see~\cite[Theorem 5]{MachiMignosi93}) that a
cellular automaton is surjective if and only if it do not admit
patterns without pre--image. The same holds for local functions on
subshifts (see~\cite[Proposition
4.1]{CeccheriniFiorenziScarabotti04}). Moreover, we
proved~\cite[Proposition 4.2]{CeccheriniFiorenziScarabotti04} that
the existence of two mutually erasable patterns is equivalent the
\emph{pre--injectivity} of the cellular automaton. This latter
property has been introduced by Gromov~\cite{Gromov99} and
corresponds to the injectivity of the automaton on the configuration
differing only on a finite set. Hence, the GOE theorem could be
restated as the equivalence between pre--injectivity and
surjectivity of a cellular automaton.

GOE--like theorems could be investigated even more generally for
local functions defined on shift spaces. Obviously, whenever the GOE
theorem (or at least Myhill's implication) holds, we have the
surjunctivity of the local function. In this connection, the GOE
theorem has been proved for one--dimensional irreducible shifts of
finite type in~\cite{Fiorenzi00}. Moreover, we prove Myhill's
implication for irreducible sofic shift. In~\cite{Fiorenzi03} we
proved that the GOE theorem for strongly irreducible shifts of
finite type on an amenable group. Finally, we proved
in~\cite{Fiorenzi04} that Myhill's implication holds for
semi--strongly irreducible shift of finite type on a group of
nonexponential growth. In all these cases the surjunctivity of a
local function holds.

\section{The $n$--dimensional case}\label{The n-dimensional case}
In this section we focus on the density of periodic configurations
for an $n$--dimensional shift. In the one--dimensional case we prove
this density for an irreducible shift of finite type of and hence, a
sofic shift being the image under a local map of a shift of finite
type, for an irreducible sofic shift. The situation in the
two--dimensional case is deeply different: there are counterexamples
of mixing shifts of finite type $X$ for which the set
$\mathcal{P}(X)$ is not dense.

We conclude this section by listing some well--known decision
problems for $n$--dimensional shifts proving that in the special
case of a one--dimensional shift they can be solved. More generally
they can be solved for the class of group shifts using some results
due to Wang~\cite{Wang61} and Kitchens and
Schmidt~\cite{KitchensSchmidt88}.

\begin{proposition}\label{dense in finite type in Z}
If $X \subseteq \mathcal{A}^\mathbb{Z}$ is an irreducible shift of
finite type, then $\mathcal{P}(X)$ is dense in X.
\end{proposition}

\begin{proof} Suppose that $X$ has memory $M$. Let $c
\in X$ and let $u_n = c_{|[-n,n]}$. Fix $w \in \mathcal{L}(X)$ with
$|w| = M$, the shift $X$ being irreducible, there exist two words
$v_n, w_n \in \mathcal{L}(X)$ such that
$$w \ v_n u_n w_n \ w \in \mathcal{L}(X).$$
Let $c_n$ be the periodic configuration
$$\dots w \ v_n u_n w_n \ w \ v_n u_n w_n \ w \dots =
\overline{w \ v_n u_n w_n}\ .$$
By Proposition~\ref{pasting in Z}, we have that $c_n \in X$.
Moreover ${c_n}_{|[-n,n]}$ = $c_{|[-n,n]}$ and hence $\lim_{n
\rightarrow \infty} c_n$ = $c$. \end{proof}

\begin{corollary}
If $X \subseteq \mathcal{A}^\mathbb{Z}$ is an irreducible sofic
shift, then $\mathcal{P}(X)$ is dense in X.
\end{corollary}

\begin{proof}
By Corollary \ref{irreducible sofic}, we have that every irreducible
sofic shift is the image under a local map of an irreducible shift
of finite type. Hence propositions~\ref{dense in finite type in Z}
and~\ref{dense} apply.
\end{proof}

\begin{counterexample}
{\rm The finite type condition does not imply, in general, the
density of the periodic configurations of a shift.}
\end{counterexample}

\begin{proof}
Let $\mathcal{A} = \{ 0,1 \}$ and let $X$ be the shift of finite
type with set of forbidden blocks $\mathcal{F} = \{ 01 \}$. Then the
elements of $X$ are the configurations $\bar 0$, $\bar 1$ and the
configurations of the type $\dots \ 111111000000\ \dots$. Clearly
$X$ is not irreducible because there are no words $u \in
\mathcal{L}(X)$ such that $0 u 1 \in \mathcal{L}(X)$. In this shift
we have $\mathcal{P}(X) = \{ \bar 0, \bar 1 \}$ which is closed (and
so not dense) in $X$.
\end{proof}
Notice that, for this shift, a local selfmapping is injective if and
only if it is surjective and hence surjunctivity holds even if the
set of periodic configurations is not dense.

\medskip

\noindent {\bf Remark.} If $X$ is a subshift of
$\mathcal{A}^\mathbb{Z}$, it is always possible to define an
irreducible subshift $\bar X$ of $\mathcal{A}^{\mathbb{Z}^2}$
consisting of copies of $X$. More precisely, a configuration $c$
belongs to $\bar X$ if and only if each horizontal line of $c$ (i.e.
the bi--infinite word $(c_{|(z,t)})_{z \in \mathbb{Z}}$, for each
fixed $t \in \mathbb{Z}$), belongs to $X$. Hence $X$ and $\bar X$
have the same set of forbidden blocks and it is obvious that the
shift $\bar X$ is of finite type if $X$ is. The irreducibility of
$\bar X$ can be easily seen: given two blocks of the shift, it
suffices to translate one of them in the vertical direction in such
a way that the supports are far enough.

\begin{counterexample}\label{finite type and irreducible in Z^2 - no dense}
{\rm The density of periodic configurations does not hold, in
general, for two--dimensional irreducible shifts of finite type.}
\end{counterexample}

\begin{proof}
Let $\bar X$ be the shift over the alphabet $\mathcal{A} = \{ 0,1
\}$ generated by the shift $X$ of the previous counterexample. Then
$\bar X$ is irreducible and of finite type. The set
$\mathcal{P}(\bar X)$ is in this case contained in the set of all
those configurations assuming constant value at each horizontal
line. It is then clear that a configuration assuming for example the
value $1$ at $(0,0)$ and $0$ at $(1,0)$, cannot be approximated with
any sequence of periodic configurations.
\end{proof}

Corollary \ref{dense in group shift} gives an answer to the problems
arising from this counterexample.

\medskip

Even if we strengthen the irreducibility hypothesis by assuming that
the shift is mixing, there are examples of two--dimensional mixing
shifts of finite type and local selfmappings which are injective and
not surjective (see \cite{Weiss00}).

\subsection{Further decision problems}\label{Further decision problems}
We conclude this section with some other decision problems arising
in the case of $n$--dimensional subshifts of finite type.

\begin{itemize}

\item[--] The \emph{tiling problem}: given a finite list $\mathcal{F}$ of
forbidden blocks is $\mathsf{X}_\mathcal{F}$ empty or non--empty? In
fact the tiling problem is an equivalent formulation of the
\emph{domino problem}, proposed by Wang \cite{Wang61}.

\item[--]
A problem strictly related with this latter is the following: given
a finite list $\mathcal{F}$ of forbidden blocks, is there a periodic
configuration in $\mathsf{X}_\mathcal{F}$?

\item[--]
Given a finite list $\mathcal{F}$ of forbidden blocks, are the
periodic configurations dense in $\mathsf{X}_\mathcal{F}$?

\item[--]
The \emph{extension problem}: given a finite list $\mathcal{F}$ of
forbidden blocks and given an \emph{allowable} block (that is a
block in which does not appear any forbidden block), is there a
configuration in $\mathsf{X}_\mathcal{F}$ in which it appears?
Clearly a positive answer to the extension problem would imply a
positive answer to the tiling problem.
\end{itemize}

Now we prove that the answers for subshifts of finite type of
$\mathcal{A}^\mathbb{Z}$ are all positive: there are algorithms to
decide, the tiling and the extension problems and there is an
algorithm to decide whether or not the periodic configurations are
dense in $X$. In order to see the first two algorithms, consider,
more generally, a sofic shift. If $\mathsf{A}$ is a finite automaton
accepting $X$ (and we may assume that the underlying graph is
essential), it can be easily seen that $X$ is non--empty if and only
if it exists a cycle on the graph. Hence the shift is non--empty if
and only if it contains a periodic configuration. On the other hand,
the language of $X$ is the language accepted by $\mathsf{A}$. Hence
an allowable word is a word of the language if and only if it is
accepted by $\mathsf{A}$.

To establish the density of the periodic configurations, suppose
that $X$ is of finite type with memory $M$. One has that
$\mathcal{P}(X)$ is dense in $X$ if and only if
$\mathcal{P}(X^{[M+1]})$ is dense in $X^{[M+1]}$. The shift
$X^{[M+1]}$ is an edge shift accepted by the graph $\mathsf{G}$
constructed in Section~\ref{Sofic shifts} and hence the set
$\mathcal{P}(X^{[M+1]})$ is dense in $X^{[M]}$ if and only if each
edge of $\mathsf{G}$ is contained in a strictly connected component
of $\mathsf{G}$, that is if the graph $\mathsf{G}$ has no edges
connecting two different connected components.

For the subshifts of finite type of $\mathcal{A}^{\mathbb{Z}^2}$ the
answers are quite different. In this setting Berger proved in
\cite{Berger66} the existence of a non--empty shift of finite type
containing no periodic configurations and the undecidability of the
tiling problem. Sufficient conditions to the decidability of tiling
and extension problems are the following.

\begin{theorem}[Wang \cite{Wang61}]
If every non--empty subshift of finite type of
$\mathcal{A}^{\mathbb{Z}^2}$ contains a periodic configuration then
there is an algorithm to decide the tiling problem.
\end{theorem}

\begin{theorem}[Kitchens and Schmidt \cite{KitchensSchmidt88}]
If every subshift of finite type of $\mathcal{A}^{\mathbb{Z}^2}$ has
dense periodic configurations then there is an algorithm to decide
the extension problem.
\end{theorem}

The following result is a consequence of these facts and of
Corollary~\ref{dense in group shift}.

\begin{corollary}
If $X \leq \mathcal{A}^{\mathbb{Z}^2}$ is a group shift, then the
tiling and extension problems are decidable for $X$.
\end{corollary}

\bibliographystyle{siam}
\bibliography{Fiorenzi09}
\end{document}